\documentclass[submission,copyright]{eptcs}
\providecommand{\event}{CASSTING'16} 

\usepackage{microtype}
\usepackage{amsmath,amsfonts,amssymb,xspace,amsthm}
\usepackage{color}
\usepackage{wrapfig}
\usepackage[ruled,boxed,vlined,linesnumbered]{algorithm2e}
\usepackage{paralist}
\usepackage{xparse}
\usepackage{enumitem}
\setlist{nosep}
\usepackage{hhline}

\providecommand{\doi}[1]{doi: \href{http://dx.doi.org/#1}{\nolinkurl{#1}}}

\usepackage{tikz}
\usetikzlibrary{arrows,shapes,decorations,automata,backgrounds,positioning}

\tikzstyle{player1}=[draw,rounded rectangle, minimum size=5mm]
\tikzstyle{player2}=[draw,rectangle,minimum size=5mm]
\tikzstyle{widget}=[draw,ellipse,dashed,minimum size=6mm]
\tikzset{every loop/.style={looseness=7}}

\renewcommand\geq{\geqslant}
\renewcommand\leq{\leqslant}
\newcommand\Z{\mathbb Z}

\newcommand\N{\mathbb N}

\newcommand\tuple[1]{\langle #1 \rangle}
\def\bar#1{\ensuremath{\overline{#1}}}

\newcommand\vertices{V}

\newcommand\edges{E}

\newcommand\edgeweights{\omega}
\newcommand\game{G}
\newcommand\Payoff{\mathbf{P}}
\ProvideDocumentCommand{\gameEx}{o}{\IfNoValueTF{#1}{\tuple{\vertices,\edges,\edgeweights,\Payoff}}{\tuple{\vertices,\edges,\edgeweights,#1}}} 

\ProvideDocumentCommand{\RPayoff}{o}{\IfNoValueTF{#1}{\mathbf{RP}}{#1\text{-}\mathbf{RP}}}
 
\newcommand\Actions{A}
\newcommand\Next{\textit{Next}}
\newcommand\Targets{F}

\newcommand\graphEx{\tuple{\vertices,\Targets,(\Actions_i)_{i\leq N},\edges,\Next,(\edgeweights_i)_{i\leq N}}}

\newcommand\play{\pi}

\newcommand\strategy{\sigma}

\newcommand\outcomes{\mathsf{Play}}

\newcommand\profStrat{\vec \strategy}

\newcommand\cost{\textit{cost}}

\newcommand\TP{\textnormal{\textbf{TP}}}

\ProvideDocumentCommand{\MCR}{o}{\IfNoValueTF{#1}{\mathbf{MCR}}{#1\text{-}\mathbf{MCR}}}
\ProvideDocumentCommand{\RDP}{o}{\IfNoValueTF{#1}{\mathbf{RDP}}{#1\text{-}\mathbf{RDP}}}
\ProvideDocumentCommand{\RMP}{o}{\IfNoValueTF{#1}{\mathbf{RMP}}{#1\text{-}\mathbf{RMP}}}

\newcommand\target{\textnormal{\texttt{t}}}

\usepackage{color}

\newtheorem{proposition}{Proposition}
\newtheorem{lemma}{Lemma}
\newtheorem{theorem}{Theorem}

\theoremstyle{definition}
\newtheorem{example}{Example}

\usetikzlibrary{decorations.pathreplacing}

\title{Efficient Energy Distribution in a Smart Grid\\ using
  Multi-Player Games%
\thanks{The research leading to these results has received funding from the European Union Seventh Framework
Programme (FP7/2007-2013) under Grant Agreement n°601148 (CASSTING).}}

\author{
Thomas Brihaye
\institute{UMONS\\ Mons, Belgium}
\email{thomas.brihaye@umons.ac.be}
\and
Amit Kumar Dhar
\institute{IIITA\\ Allahabad, India}
\email{amitdhar@iiita.ac.in}
\and
Gilles Geeraerts
\institute{ULB\\ Brussels, Belgium}
\email{gigeerae@ulb.ac.be}
\and
Axel Haddad
\institute{UMONS\\ Mons, Belgium}
\email{axel.haddad@umons.ac.be}
\and
Benjamin Monmege
\institute{LIF, Aix-Marseille Universit\'e, CNRS\\ Marseille, France}
\email{benjamin.monmege@lif.univ-mrs.fr}
}

\def\titlerunning{Efficient Energy Distribution in a Smart Grid using
  Multi-Player Games} 

\def\authorrunning{T. Brihaye, A. K. Dhar, G. Geeraerts, \\ A. Haddad, 
  B. Monmege}

\begin{document}
\maketitle
\def\authorrunning{T. Brihaye, A. K. Dhar, G. Geeraerts, A. Haddad, 
  B. Monmege}
\markboth{\hfill\titlerunning}{\authorrunning\hfill} 

\begin{abstract}
  Algorithms and models based on game theory have nowadays become
  prominent techniques for the design of digital controllers for
  critical systems. Indeed, such techniques enable \emph{automatic
    synthesis}: given a model of the environment and a property that
  the controller must enforce, those techniques automatically produce
  a correct controller, when it exists. In the present paper, we
  consider a class of concurrent, weighted, multi-player games that are
  well-suited to model and study the interactions of several agents
  who are competing for some measurable resources like energy. 
  We prove that a subclass of those games always admit a
  Nash equilibrium, i.e.\ a situation in which all players play in such
  a way that they have no incentive to deviate. Moreover, the
  strategies yielding those Nash equilibria have a special structure:
  when one of the agents deviate from the equilibrium, all the others
  form a coalition that will enforce a retaliation mechanism that
  punishes the deviant agent.
  We apply those results to a real-life case study in which several
  smart houses that produce their own energy with solar panels, and
  can share this energy among them in micro-grid, must distribute the
  use of this energy along the day in order to avoid consuming
  electricity that must be bought from the global grid. We demonstrate
  that our theory allows one to synthesise an efficient controller for
  these houses: using penalties to be paid in the utility bill as an
  incentive, we force the houses to follow a pre-computed schedule
  that maximises the proportion of the locally produced energy that
  is consumed.
\end{abstract}

\section{Introduction}

A recent and well-established research direction in the field of the
design of digital controller for critical systems consists in applying
concepts, models and algorithms borrowed from \emph{game theory} to
perform \emph{automatic synthesis (construction)} of correct
controllers. The contributions of the present article are part of this
research effort. In the setting of automatic synthesis, the controller
we want to build is a \emph{player} (using the vocabulary of game
theory), and the specification that the controller should satisfy is
cast as a game objective that the controller player should enforce at
all times, regardless of the behaviour of the environment. The
environment itself is modeled as another player (or a set of
players). Thus, computing automatically a correct and formally
validated (with respect to the specification) controller boils down to computing a
\emph{winning strategy} for the controller, i.e.\ a strategy that
ensures this player to win the game whatever the other players play.
%
%
System synthesis through a game-based approach has nowadays reached a
fair amount of maturity, in particular thanks to the development of
several tools (such as UppAal TiGa~\cite{BCDFLL07}, UppAal
Stratego~\cite{DJLMT15} and Prism Games~\cite{CFK+13}) that have been
applied successfully to real life case studies (see for
instance~\cite{DGJLR10,CKSW13,LMT15}).
Until recently, however, the research has mainly focused on two-player
games, where the players (the system and the environment) have
antagonistic objectives.  This approach allows one to model and reason
on centralised control only.  Although multi-player games have been
studied from an algorithmic point of view, strategic forms of those
games have been mainly considered, and the study of multi-player games
played on graphs---the kind of model we need in our setting---is
relatively recent.  This research direction is part of the CASSTING
project\footnote{\url{http://www.cassting-project.eu/}}, whose aim is
to propose new techniques for the synthesis of \emph{collective
  adaptive systems}. Such systems are decentralised and consist of
several modules/agents interacting with each other.  While the idea of
using games remains, using multi-player games for synthesis of
collective adaptive systems represents a huge leap in game theory for
synthesis.  Indeed, in adversarial games the main goal is to find
winning (or optimal in a quantitative setting) strategies, whereas in
multi-player games, one wants to synthesise controllers by computing
equilibria (such as Nash equilibria \cite{Nash}) characterising an
adequate behaviour of each agent.  \medskip

In this article, we consider a class of \emph{quantitative
  multi-player games} that are well-suited to model systems where a
quantity grows or decreases along the plays (this quantity can model
some energy level, economic gain/loss, or any other measurable
resource). More precisely, the game models a multi-state system where
the players choose their actions concurrently (at the same time), and
the next state is a function of the current state and the players'
actions. Going from one state to another can result in a positive or
negative cost for the players.  One can give two semantics to these
games, either an infinite horizon semantics where the plays are
infinite and the players want to minimize the limit
(inferior/superior) of the partial sum of the costs; or a finite
horizon semantics where the goal of each player is to reach some
target state, and minimize the sum of the costs paid before reaching
the target.  In the following, we focus on the latter semantics, more
fitted to the case study, we will be interested in.

We start by establishing some properties of these games.  Although
there may not always exist Nash equilibria in these games, we describe
a subclass in which there always exist some. First we observe that
when several players play at the same time concurrently, one can
encounter a situation similar to the rock-paper-scissor game, in which
there is no (pure) Nash equilibrium. However, even in a turn-based
game (i.e.\ a game in which, in each state, only one player is in
charge of choosing the next state), there may not always exist a Nash
equilibrium.  More precisely, we show that---unlike in some other
classes of games---it is possible that each player cannot
independently guarantee his cost to decrease arbitrarily, while a
coalition of all players can achieve this goal.  We then show that
this is the only situation that prevents Nash equilibria from existing
in those games: we prove that, when the cost of any play is bounded from
below by a fixed threshold, then a Nash equilibrium exists in the
game.
\medskip

We demonstrate the applicability of this theory in a practical
situation.  We consider a case study introduced by an industrial
partner of the CASSTING project, and model it in a game formalism in
order to build a controller fulfilling a specific set of goals. The
case study consists of a local grid of eight houses equipped with
solar panels.  The solar panels produce different amounts of energy
during the day. When they need to consume energy, the houses can
either rely on energy produced by the solar panels (their own or one
of the seven other houses') or buy it from the global grid. The aim
of the case study is to minimise the use of energy bought from the
global grid as a whole, while preserving the incentive of each house
to share the energy produced by their solar panel with others, if not
used directly by them.  We assume that the energy produced by the
solar panel has to be used within a small interval of time and can not
be stored (a provision for storage of energy could be added at little
increase of modelling complexity).  Concretely, we want to generate a
controller producing a schedule of the different tasks of the houses,
such that each house has no incentive to deviate from this
schedule. For this, we assume that there are two types of
controllers. One global controller which has information about all the
houses, their requirements and their production. Also, there are local
controllers in every house communicating with the global controller
and controlling the tasks that take place in this house. Local
controllers have no information about the consumption or production of
the other houses: they are only aware of the energy produced by their
own solar panels and the energy requirement of the house at any
specific interval of the day. In our experiments, this schedule is
computed as a strategy in a multi-player concurrent game that:
\begin{inparaenum}[1.]
\item minimises the energy bought from outside; and
\item minimises the bill to be paid by each house.
\end{inparaenum}
We also assume that the houses are not bound to follow the schedule
and can deviate from it. However, such deviations could lead to a
severe increase in the overal consumption from the global grid (if,
for instance, a house decides to use its own energy locally instead of
injecting it on the local grid as prescribed by the schedule, then the
total amount of energy available on the local grid might be too low,
and energy might have to be bought from the global grid).  For this
reason, we devise proper incentive and a penalty mechanism ensuring
that the houses would not have any interest in deviating from the proposed schedule.


\section{Theoretical background}

We first introduce the class of multi-player games we are interested
in. We fix a number $N$ of players and let $\{1,\ldots,N\}$ be the set
of players. A \emph{concurrent min-cost reachability (MCR) game} is a
tuple $\graphEx$ where $\vertices$ is a finite set of vertices
partitioned into the sets $\vertices_1,\ldots,\vertices_N$,
$\Targets\subset \vertices$ is a subset of vertices called
\emph{targets}, 
for every
vertex $v\in\vertices$,
$\Actions_i$ is a finite set of actions for player
$i$, $\edges\subseteq \vertices\times \vertices$ is a set of
\emph{directed edges}, the set of successors of $v$ by $\edges$ is
denoted by $\edges(v) = \{v'\in\vertices\mid (v,v')\in \edges\}$,
$\Next: \vertices \times \prod_i A_i \rightarrow \vertices$ is a
mapping such that for all $v$ and $a_1,\ldots,a_N$,
$\Next(v,a_1,\ldots,a_N)\in \edges(v)$,
$\edgeweights_i\colon \edges \to \Z$ is the \emph{weight function} for
player~$i$, associating an integer weight with each edge. 
Without loss of generality, we assume that every graph is
deadlock-free, i.e.\ for all vertices~$v$, $\edges(v)\neq\emptyset$.
In the following we let $\Actions=\prod_{i\leq N}\Actions_i$.
A \emph{finite play} is a finite sequence of vertices
$\play=v_0v_1\cdots v_k$ such that for all $0\leq i<k$,
$(v_i,v_{i+1})\in \edges$. 
A \emph{play} is an infinite sequence of
vertices $\play = v_0v_1\cdots$ such that every finite prefix
$v_0\cdots v_k$, denoted by $\play[k]$, is a finite play.

The total-payoff of a finite play $\play=v_0 v_1 \cdots v_k$ for
player~$i$ is obtained by summing up the weights along~$\play$, i.e.\
$\TP_i(\play) = \sum_{\ell=0}^{k-1}
\edgeweights_i(v_\ell,v_{\ell+1})$.
The total-payoff of a play $\play$ is obtained by taking the limit
over the partial sums, i.e.\
$\TP_i(\play) = \liminf_{k\to \infty} \TP_i(\play[k])$. The cost of a
play, $\cost_i(\play)$ is $+\infty$ if $\play$ does not visit any
target, and $\TP_i(v_0 v_1 \cdots v_\ell)$ otherwise, with $\ell$ the
least index such that $v_\ell\in \Targets$: it reflects that players
want to minimise their cost, subject to the imperative of reaching the
target as a primary objective.

A \emph{strategy} for player~$i$ is a mapping
$\strategy\colon \vertices^+ \to A_i$. 
A play or finite play
$\play = v_0v_1\cdots$ conforms to a strategy $\strategy$ of player
$i$ if for all $k$, there exists $(a_1,\cdots,a_N)\in \Actions$ such
that $a_i = \strategy(v_0,\ldots,v_k)$, and
$v_{k+1} = \Next(v,(a_1,\ldots,a_N))$.  A profile of strategies is a
tuple $(\strategy_1,\ldots,\strategy_N)$ where for all $i$,
$\strategy_i$ is a strategy of player~$i$. 
%
For all profiles of strategies
$\profStrat =(\strategy_1,\ldots,\strategy_N)$, for all vertices $v$,
we let $\outcomes(v,\profStrat)$ be the outcome of~$\profStrat$,
defined as the unique play conforming to $\strategy_i$ for all $i$,
and starting in~$v$, i.e.\ the play $v_0 v_1 \cdots$ such that $v_0=v$
and for all $\ell$, $v_{\ell+1} = \Next(v_\ell, (a_1,\ldots,a_N))$
where $a_i=\strategy_i(v_1\cdots v_\ell)$.

A profile of strategies $\profStrat =(\strategy_1,\ldots,\strategy_N)$
is a (pure) \emph{Nash equilibrium} from vertex $v$, if for all
players $i$, and for all strategies $\strategy'_i$,
$\cost_i(\outcomes(v,(\strategy_1,\ldots,\strategy'_i, \ldots,
\strategy_N))) \geq \cost_i(\outcomes(v,\profStrat))$.
Observe that we assume that the objective of every player is to minimise
its cost; thus, intuitively, a profile of strategies is a Nash
equilibrium if no player has an incentive to deviate.

We say that a vertex $v$ \emph{belongs} to some player $i$ if he is
the only one to choose the next vertex, i.e.\ for all pairs of actions
$(a_1,\ldots,a_N)$ and $(a'_1,\ldots,a'_N)$, if $a_i = a'_i$ then
$\Next(v,(a_1,\ldots,a_N)) = \Next(v,(a'_1,\ldots,a'_N))$.  A game is
said to be \emph{turned-based} if each vertex belongs to some
player. When considering turned-based games, instead of actions, we
say that the players to whom the current vertex belongs chooses
directly the next vertex.

\subsection{Nash equilibria do not always exist\ldots}

A natural question is the existence of Nash equilibria in the min-cost
reachability games we have just defined. In order to understand
precisely what are the conditions that prevent the existence of Nash
equilibria, we present some examples of min-cost reachability games in
which we can show that no such equilibria exist. We also recall
previous results identifying classes of games where such equilibria
are guaranteed to exist. We start with a game that is \emph{not}
turn-based and admits \emph{no Nash equilibria}.

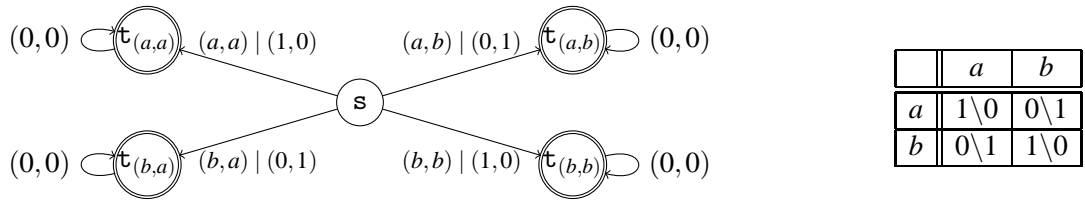
\begin{figure}[t]
  \centering
  \begin{tikzpicture}[node distance=4cm]
    \node[draw,circle] (s) {$\texttt s$};
    \node[draw,circle,double,above left of=s,yshift=-2cm] (aa) {\makebox[10pt][c]{$\texttt t_{(a,a)}$}};
    \node[draw,circle,double,above right of=s,yshift=-2cm] (ab) {\makebox[10pt][c]{$\texttt t_{(a,b)}$}};
    \node[draw,circle,double,below left of=s,yshift=2cm] (ba) {\makebox[10pt][c]{$\texttt t_{(b,a)}$}};
    \node[draw,circle,double,below right of=s,yshift=2cm] (bb) {\makebox[10pt][c]{$\texttt t_{(b,b)}$}};

    \draw[->] (s) -- node[above,yshift=1mm] {\footnotesize$(a,a)\mid (1,0)$} (aa);
    \draw[->] (s) -- node[below,yshift=-1mm] {\footnotesize$(b,b)\mid (1,0)$} (bb);
    \draw[->] (s) -- node[above,yshift=1mm] {\footnotesize$(a,b)\mid (0,1)$} (ab);
    \draw[->] (s) -- node[below,yshift=-1mm] {\footnotesize$(b,a)\mid
      (0,1)$} (ba);
    \draw[->] (aa) edge[loop left] node[left] {$(0,0)$} (aa);
    \draw[->] (ab) edge[loop right] node[right] {$(0,0)$} (ab);
    \draw[->] (ba) edge[loop left] node[left] {$(0,0)$} (ba);
    \draw[->] (bb) edge[loop right] node[right] {$(0,0)$} (bb);
 \end{tikzpicture}
 \hspace{2cm}
 \raisebox{3mm}{\begin{tikzpicture}
   \node () {$
     \begin{array}{|c||c|c|}\hline
        & a & b \\\hline\hline
       a & 1\backslash 0 & 0\backslash 1  \\\hline
       b & 0\backslash 1 & 1\backslash 0  \\\hline
     \end{array}$};
 \end{tikzpicture}}
  \caption{A game without pure Nash equilibrium representing the
    strategic game described by the matrix on the right: we do not
    depict the actions on the loops over targets for conciseness.}
  \label{fig:noNash}
\end{figure}

\begin{example}
  Consider the game in \figurename~\ref{fig:noNash} with four target
  vertices $\texttt{t}_{(a,a)}$, $\texttt{t}_{(a,b)}$,
  $\texttt{t}_{(b,a)}$, $\texttt{t}_{(b,b)}$, and one additional
  vertex $\texttt{s}$. Assume that there are two players, and each has
  two possible actions: $a$ and $b$.  From $\texttt{s}$ the pair of
  actions $(\alpha_1,\alpha_2)$ leads to
  $\texttt{t}_{(\alpha_1,\alpha_2)}$.  If both players choose the same
  action, the cost for player~1 is $1$ and the one for player~2 is
  $0$; if both players choose different actions, the cost for player~1
  is $0$ and the one for player~2 is $1$. There is clearly no (pure)
  Nash equilibrium in this game from $\texttt{s}$ since for all
  profiles of strategies, either player~1 or player~2 would pay less
  with another strategy\footnote{It is however possible to find Nash
    equilibria that use randomisation (so-called mixed strategies),
    but we do not consider such objects in this work.}.
\end{example}

In a turn-based setting, one can also easily exhibit examples with no
pure Nash equilibria.

\begin{example}\label{noNE-tb}
  Consider a one player game with two vertices $v_1$ and $v_2$ where
  the latter is the only target. The set of edges is
  $\{(v_1,v_1), (v_1,v_2),(v_2,v_2)\}$, all with cost $-1$. In other
  words, from $v_1$, the player can either choose to loop, and get a
  reward (since he seeks to minimise his cost); or to reach the target
  $v_2$ (in which case the play formally continues with no influence
  on the cost). In this game, a strategy from vertex~$v_1$ can thus be
  described by the number of times he will loop on $v_1$ before going
  to $v_2$. If he never reaches $v_2$, he pays $+\infty$ which is
  clearly bad. If he loops $n$ times, a strictly better strategy would
  be to loop $n+1$ times, therefore there is no Nash equilibrium in
  that game.
\end{example}

In~\cite{Pri13,BPS13}, the authors introduce a large class of
turn-based games for which they prove that a pure Nash equilibrium
always exists.  In particular, this result can be used to show that
every turned-based min-cost reachability game with only positive costs
admits a (pure) Nash equilibrium. From Example~\ref{noNE-tb}, we
already know that when there are negative costs, this result does not
hold anymore.  In this example, the (only) player has a family of
strategies that allows him to secure a cost which is arbitrary low,
hence, the absence of Nash equilibria is not too surprising.  Let us
now exhibit a third, two-player example in which no player has a
strategy to guarantee, individually, an arbitrary low cost; but still
arbitrary low costs can be secured when the players cooperate. Again,
this phenomenon forbid the existence of Nash equilibria.


Note that here, we only look at pure Nash equilibria.  In the general
setting of mixed strategies, i.e.\ where the players pick randomly a
strategy according to some probability distribution over the set of
pure strategies, there is a Nash equilibrium in this game.  Indeed if
we let $\strategy^n$ the strategy consisting in looping $n-1$ times
around $v_1$ and then going to $v_2$ (ensuring a cost of $-n$), the
distribution consisting in picking $\strategy^n$ with probability
$\frac{6}{(\pi n)^2}$ (one can easily check that it is a distribution)
ensures an expected cost of $-\infty$.

\begin{example}\label{ex:3}
  Let $G$ be a turn-based game with two players $1$ and $2$, and three
  vertices $A$, $B$, and $C$. Vertex~$C$ is the only target. $A$ and
  $C$ belong to player~$1$ and $B$ belongs to player~$2$. The edges
  and the weight function are depicted in \figurename~\ref{fig:game}
  (e.g.\ $\edgeweights_1(A,C) = 0$ and $\edgeweights_2(A,C) = -1$).
\end{example}


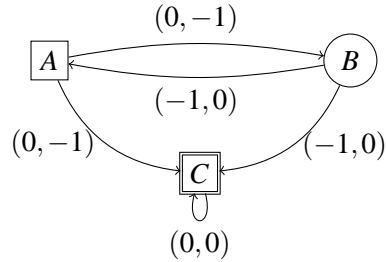
\begin{figure}\centering
\begin{tikzpicture}
\node[draw, rectangle] (A) at (-2,1) {$A$};
\node[draw, circle] (B) at (2,1) {$B$};
\node[draw, double, rectangle] (C) at (0,-0.5) {$C$};

\draw[->] (A) to[bend left=10] node[midway,above] {$(0,-1)$} (B);
\draw[->] (B) to[bend left=10] node[midway,below] {$(-1,0)$} (A);
\draw[->] (A) to[bend right] node[midway,left] {$(0,-1)$} (C);
\draw[->] (B) to[bend left] node[midway,right] {$(-1,0)$} (C);
\draw[->] (C) edge[loop below] node[midway,below] {$(0,0)$} (C);
\end{tikzpicture}
\caption{A turn-based MCR game with no Nash equilibria, but where no
  player can independently guarantee an arbitrary low cost.}
\label{fig:game}
\end{figure}

\begin{proposition}
There is no pure Nash equilibrium in the game $G$, neither from $A$ nor from $B$. 
\end{proposition}
\begin{proof}
  We do the proof for plays starting in~$A$, it is easily adapted to
  the other case. Note that the set of finite plays ending in $A$ is
  $A(BA)^*$ and the set of finite plays ending in $B$ is $(AB)^+$. Let
  $(\strategy_1,\strategy_2)$ be a profile of strategies and let
  $\play$ be its outcome. We consider several cases for $\play$: 
  \begin{enumerate}
  \item If $\play = (AB)^\omega$ then
    $\textit{Cost}_1(\play) = \textit{Cost}_2(\play) = +\infty$. Then,
    let $\strategy'_1$ be the strategy of player~$1$ defined by
    $\strategy'_1(A (BA)^n) = C$ for all $n$ and
    $\strategy'_1(\pi C) = C$ for all finite play $\pi$, then the
    outcome of $(\strategy'_1,\strategy_2)$ is $AC^\omega$ and
    $\textit{Cost}_1(AC^\omega) = 0$, which is strictly better than
    $+\infty$, and player $1$ has an incentive to deviate.
  
  \item If $\play = (AB)^n C^\omega$ for some $n$, then
    $\textit{Cost}_2(\play) = -n$. Let $\strategy'_2$ be the strategy
    obtained from $\strategy_2$ by letting $\strategy'_2((AB)^n) = A$
    and $\strategy'_2((AB)^{n+1}) = C$. One can easily check that the
    outcome of $(\strategy_1,\strategy'_2)$ is either
    $(AB)^{n+1} C^\omega$ or $(AB)^n A C^\omega$, and in both cases,
    the cost of this play for player~$2$ is $-(n+1)$ which is strictly
    better than $+\infty$, hence he has an incentive to deviate.

  \item Finally, if $\play = A (B A)^n C^\omega$ for some $n$ then
    $\textit{Cost}_1(\play) = -n$. Let $\strategy'_1$ be the strategy
    obtained from $\strategy_1$ by letting
    $\strategy'_1 (A (BA)^n) = B$ and $\strategy'_1(A(BA)^{n+1}) = C$.
    One can easily check that the outcome of
    $(\strategy'_1,\strategy_2)$ is either $(AB)^{n+1} C^\omega$ or
    $(AB)^{n+1} A C^\omega$, and in both cases, the cost of this play
    for player~$1$ is $-(n+1)$, hence, again, player~$1$ has an
    incentive to deviate.
  \end{enumerate}
  We conclude that, in all cases, one of the players has an incentive
  to deviate, hence no profile of strategies
  $(\strategy_1,\strategy_2)$ is a Nash equilibrium.
\end{proof}

\subsection{Lower-bounded set of costs}
As already outlined, the intuition behind the absence of Nash
equilibria in Examples~\ref{noNE-tb} and \ref{ex:3} is the existence
of plays with arbitrary low costs (even if these plays can not be
enforced by a single player, as shown by Example~\ref{ex:3}). We will
now show that this is indeed a necessary condition for the absence of
Nash equilibria. In other words: in a min-cost reachability game (with
arbitrary weights), \emph{if} the set of possible total-payoffs of
finite plays is bounded from below, \emph{then} a Nash equilibrium is
guaranteed to exist.

\begin{theorem}\label{thm:pureNash}
  Let $G$ be a turn-based MCR game respecting the following condition:
  for all players $i$, there exists $b_i\in \N$ such that all finite
  plays $\play$ satisfy $\TP_i(\play) \geq -b_i$. Then there exists a
  pure Nash equilibrium from all vertices in $G$.
\end{theorem}

We now prove this theorem. For that purpose, let
$G=(\vertices,\{\target\}, (\Actions_j)_{j\leq
  N},\edges,\Next,(\edgeweights_j)_{j\leq N})$
be a concurrent MCR game (we will restrict ourselves to a turn-based
game when necessary). For the sake of simplicity, we assume here that
there is a unique target $\target$ for all players, and the only
outgoing edge from~$\target$ is the loop $(\target,\target)$. Note
that the following construction would hold for multiple targets as
well.

In the following, we fix a player $i$, and we assume that there exists
$b_i\in \N$ such that for all finite plays~$\play$,
$\TP_i(\play) \geq -b_i$. We will show how to translate the game $G$
in a game $G'$ with only non-negative weights for player $i$, with a
relationship between strategies of $i$ in both games. This will in
particular preserve the existence of Nash equilibria. The game
$G'=(\vertices',\{\target\}, (\Actions'_j)_{j\leq
  N},\edges',\Next',(\edgeweights'_j)_{j\leq N})$ is built as follows:
 \begin{itemize}
 \item
   $\vertices' = \{\target\} \uplus \vertices\times
   \{-b_i,\ldots,-1,0\}$:
   we keep the negative part of the current total-payoff in the vertex
   for player $i$, and add a fresh target vertex $\target$;
\item $\Actions'=\Actions$;
\item for all $(v,v')\in \edges$ and for all $(v,c)\in \vertices'$
  with $v\neq \target$, then, letting
  $c'= \min(0,c+\edgeweights_i(v,v'))$, if $(v',c')$ is in
  $\vertices'$, the edge $e=((v,c),(v',c'))$ is in $\edges'$,
  $\edgeweights'_i(e)= \max(0,{c}+\edgeweights_i(v,v'))$ and
  $\edgeweights'_j(e) = \edgeweights_j(v,v')$ for $j\neq i$.
  Furthermore if $\Next(v,\vec a) = v'$, then
  $\Next((v,c),\vec a) = (v',c')$;
\item for all vertices $(\target,{c})\in \vertices'$, there exists an
  edge $e= ((\target,{c}),\target)$ with
  $\edgeweights'_i (e)= -b_i + {c}$ and $\edgeweights'_j(e) = 0$ for
  $j\neq i$. For all $\vec a$,
  $\Next((\target,{c}), \vec a) = \target$.
\end{itemize}

\begin{lemma}\label{lemmaBounded}
  For all finite plays $v_1 v_2 \cdots v_kv_{k+1}$ in $G$, and
  $(v_1, 0)(v_2,{c}_2) \cdots (v_k,{c}_k)$ in $G'$,
\begin{itemize}
\item there exists $j\leq k$ such that ${c}_k = \TP_i(v_j\cdots v_k)$
  in $G$ (note that if $j=k$ this is equal to 0), and in $G'$
  \[\TP_i((v_1,c_0)(v_2,{c}_2) \cdots (v_k,{c}_k))= \TP_i(v_1\cdots
  v_j)\,,\]
\item if $v_{k}\neq \target$, there exists a unique $c$ such that
  $((v_k,{c}_k),(v_{k+1},c))$ is an edge of $G'$.
 \end{itemize}
\end{lemma}
\begin{proof}
  The first point is proven by induction as for all $j$, $c_j$ is
  either equal to $0$ or to $c_{j-1}+ \edgeweights(v_{j-1},v_j)$.  The
  second point is a consequence of the first. As $c_k$ is the weight
  of a partial play ending in $v_k$,
  $c_k+ \edgeweights(v_{k},v_{k+1})$ is the weight of a partial play
  ending in $v_{k+1}$ thus
  $\min(0, c_k+ \edgeweights(v_{k},v_{k+1}) ) \geq b_i$.
\end{proof}

As a consequence, for all plays or finite plays $\pi = v_1 v_2 \cdots$
in $G$, there exists a unique play or finite play
$\bar \pi = (v_1,0) (v_2,c_2) \cdots$ in $G'$ such that if
$\Next(v_i,\vec a) = v_{i+1}$ then
$\Next'((v_i,c_i),\vec a)= (v_{i+1},c_{i+1})$. Following this, one can
map every strategy $\strategy$ in $G$ to a strategy $\bar \strategy$
of the same player in $G'$ satisfying
$\bar \strategy (\bar \pi) = \strategy(\pi)$ for all finite plays
$\pi$. Furthermore, for all strategies $\strategy$ in $G'$, there
exists a unique strategy $\strategy^\star$ for the same player in $G$
such that $\strategy^\star (\pi) = \strategy(\bar \pi)$, for all
finite plays $\pi$.

\begin{proposition}
  \begin{enumerate}
  \item Let $\vec \strategy$ be a profile of strategies in $G$ and
    ${\vec {\bar \strategy}}$ its image in $G'$. Then, for each
    initial vertex $v$,
    $\cost_i(\outcomes((v,0),{\vec {\bar \strategy}})) =
    \cost_i(\outcomes(v,{\vec \strategy})) - b_i$
    and
    $\cost_j(\outcomes((v,0),{\vec {\bar \strategy}})) =
    \cost_j(\outcomes(v,{\vec \strategy}))$ for all $j\neq i$.
  \item Let $\vec \strategy$ be a profile of strategies in $G'$ and
    ${\vec {\strategy^\star}}$ its image in $G$. Then, for each
    initial vertex $v$,
    $\cost_i(\outcomes(v,{\vec \strategy})) =
    \cost_i(\outcomes((v,0),\vec {\strategy^\star})) - b_i$
    and
    $\cost_j(\outcomes(v,{\vec \strategy})) =
    \cost_j(\outcomes((v,0),\vec {\strategy^\star}))$
    for all $j\neq i$.
  \end{enumerate}
\end{proposition}
\begin{proof} 
  We prove here only the first item, the proof of the second being
  similar. Let $\pi=\outcomes(v_0,\vec \strategy) = v_0 v_1 \cdots$.
  As a consequence of the above remarks,
  ${\bar \pi} = \outcomes((v_0,0),\vec {\bar \strategy})$. Therefore,
  if $\pi$ does not reach a target then neither does $\bar \pi$, thus
  $\cost_i(\outcomes((v,0),{\vec {\bar \strategy}})) =
  \cost_i(\outcomes(v,{\vec \strategy})) = +\infty$.
  Assume now that $\pi = v_0 \cdots v_k \target \target \cdots$.  By
  definition,
  $\bar \pi = (v_0,0) \cdots (v_k,c_k) (\target,c_{k+1}) \target
  \cdots$.
  From Lemma~\ref{lemmaBounded}, there exists $j\leq k+1$ such that
  $c_{k+1} = \TP_i(v_j \cdots v_{k+1})$ and
  $\TP_i((v_1,c_0)(v_2,{c}_2) \cdots (\target,{c}_{k+1}))=
  \TP_i(v_1\cdots v_j)$.
  Thus
  $\cost_i(\bar \pi) = \TP_i(v_1\cdots v_j) + c_{k+1} -b_i =
  \cost_i(\pi) - b_i$.
  It is immediate that for all $j\neq i$,
  $\cost_j(\bar \pi) = \cost_j(\pi)$.
 \end{proof}

 As a consequence, there is a Nash equilibrium from $v$ in $G$ if and
 only if there is a Nash equilibrium from $(v,0)$ in $G'$. 
 Note that we could not have reached this result simply by shifting the weights above $0$, as we need a device to simulate the fact that the sum of the weights
 can also decrease during the computation.
 
 By applying
 this construction for all players, we can show that if there exists a
 lower bound on the total-payoff of the finite plays for all players
 (i.e.\ the hypothesis of Theorem~\ref{thm:pureNash} is fulfilled),
 then one can construct a game $G'$ with only non-negative weights
 such that there is a Nash equilibrium in $G'$ if and only if there
 exists a Nash equilibrium in $G$. From the fact that all turn-based
 MCR-games with non-negative weights have a pure Nash
 equilibrium~\cite{Pri13}, one obtains the proof of
 Theorem~\ref{thm:pureNash} in the special case where $G$ is
 turn-based (since $G'$ is also turn-based in this case).

\subsection{Characterising Nash equilibria outcomes}\label{folk}

In this section we present a very handy characterisation of Nash
equilibria that has recently been used in several
works~\cite{KLSTC12,BPS13}.  This intuitive characterisation, in the spirit of the folk theorem for repeated games, has been
formally stated in~\cite{H16}, and a more general and more involved
version can be found in~\cite{BBMU15}. Roughly speaking, this
characterisation amounts to reducing the computation of a Nash
equilibrium in an $n$-player games to the computation of the Nash
equilibria in $n$ versions of 2-player games, obtained by letting each
player of the original game play against a coalition of all the other
players. The usefulness of this technique stems from the fact that
2-player (zero-sum turn based) games have been widely studied and there
exists many algorithms and tools to solve them.

Thus, we first introduce a variant of our games, called \emph{2-player
  zero-sum MCR games}.  Such a game is very similar to a 2-player MCR
game, the only difference is while one of the players has the same
objective as in a standard game (i.e.\ reaching a target while
minimising its cost), the second player has a completely antagonistic
goal, i.e.\ either avoiding the targets or maximising the cost for the
first player.  In those games, we are interested at the infimum cost
that the first player can ensure, that we call the \emph{value} of
the game, denoted $value(G)$ for a game $G$, supposing that an initial
vertex is described in $G$.

Then, we introduce the notion of \emph{coalition games}.  Given an
MCR-game $G$, a player $i$, and a finite play~$\pi$ ending in vertex
$v$, the coalition game $G_{i,\pi}$ is the 2-player zero-sum
turned-based MCR game played on $G$ from $v$, where $i$ is the player
who wants to reach the target while minimising his costs; and his
adversary, denoted $-i$, has the same actions as the product of all
players except $i$, and its goal is antagonistic to the one of $i$.
Furthermore, to obtain a turn-based game, we assume that $-i$ chooses
its actions before $i$ (see~\cite{H16} for a formal definition).  It
matches the intuition that player $-i$ is a coalition of all players
but $i$, and whose goal is to make $i$ pay the most.

The characterisation we present works in the case of
\emph{action-visible} MCR-game, i.e.\ in a game where we assume that
the players know the actions that have been played by everyone. A
similar result holds in the general case~\cite{BBMU15}, but we do not
need it here as the games introduced in the next section are all
action-visible.  More precisely to be action visible, we assume that
for all $v,v'$, there exists at most one vector of actions $\vec a$
such that $\Next(v,\vec a) = v'$.

Now, assume an action-visible MCR-game, a player $i$ and a play
$\pi = v_1 v_2 \cdots$.  An \emph{$i$-deviation} from $\pi$ is a
finite play $\pi'=v_1 \cdots v_\ell v'$ such that if we let $\vec a$
and $\vec {a'}$ be the vectors of actions satisfying
$\Next(v_\ell,\vec a)=v_{\ell+1}$ and $\Next(v_\ell,\vec {a'})=v'$,
then $a_i\neq a'_i$ and $a_j=a'_j$ for all $j\neq i$.  Intuitively,
this means that all players have agreed to play according to $\pi$,
and an $i$-deviation describes a finite play in which player $i$ has
betrayed the other players.  One can now state the theorem
from~\cite{H16}.
\begin{theorem}\label{the:Axel}
  Let $G$ be an action-visible MCR-game and $\pi = v_1 v_2\cdots$.
  Then $\pi$ is the outcome of a Nash equilibrium, if and only if, for
  all players $i$ and for all $i$-deviations
  $\pi' = v_1 \cdots v_\ell v'$:
\[\cost_i(\pi) \leq \cost_i(\pi') + \textit{value}(G_{i,\pi'}).\]
\end{theorem}

In other words, this theorem allows us to say that a Nash equilibrium can be
characterised by
\begin{inparaenum}[(i)]\item 
  a play that all players agree to follow; and 
\item a set of coalition strategies that the faithful players will
  apply in retaliation if one player deviates.
\end{inparaenum}
It also provides a heuristic to construct a Nash equilibrium by
solving a sequence of $2$-player zero-sum turn-based games. It works
as follows:
\begin{inparaenum}[(i)]
\item compute for each player $i$, a strategy~$\strategy_i$ ensuring
  the least possible cost against a coalition of all other players;
\item consider the outcome $\pi$ of the profile
  $\vec \strategy=(\strategy_1,\ldots,\strategy_N)$; 
\item check that all deviations satisfy the above property; and
\item if it is the case, compute coalition strategies in case of a
  deviation.
\end{inparaenum}

Note that this construction does not always work, as the outcome $\pi$
could fail to satisfy the property of Theorem~\ref{the:Axel}, but it
has been proved to always succeed in many known classes of
games~\cite{BPS13,Pri13}.  We use this technique in the case study, as
described in the following section.


\section{Modelling}
In this section, we model the CASSTING
case study described in the introduction via a concurrent MCR
game. Recall that the problem consists of:
\begin{inparaenum}
\item a group of houses $H=\{H_1,H_2,\cdots,H_N\}$ in a cluster with
  solar panels;
\item a production function giving the (probabilistic) distribution of
  amount of solar energy produced throughout the day;
\item a list of tasks that the houses need to perform throughout the
  day.
\end{inparaenum}
For the sake of modelling, we divide each day into 15 minutes
intervals. Thus, we have 96 time intervals. We take the production
function $prod:[1,96]\rightarrow \Z$ giving the production of
energy from each house at any given time interval within the day. We
assume that for each house there is a local controller and one global
controller for all the houses together. The local controller submits a
list of tasks along with favourable time interval (deadline) of the
day for the task, and receives a schedule indicating which task to
perform when. The global controller gathers a list of tasks from
houses and computes a schedule of the tasks so that it achieves the
goal; it also sends the respective schedules to the local
controllers. Note that the local controllers have no information about
other houses and their consumption.

\subsection{Tasks}
We assume that, at the beginning of the day, each house submits a list
of tasks that should be performed at preferred interval of time. Let
the list of tasks be given as $T=\{T_1,T_2,\ldots,T_m\}$ for some
$m$. Without loss of generality, we assume that each task can be
completed within only one time interval. The energy consumed by a task
during each interval is given by the function
$E_T:T\rightarrow\Z$. The task list submitted by each house $H_i$ is
of the form
$TL_i=(\langle t_1,I_1\rangle,\langle t_2,I_2\rangle,\ldots,\langle
t_k,I_k\rangle)$
where $t_i\in T$ and $I_i$ is an interval of $[1,96]$ for all
$i\in \{1,\ldots,k\}$. For the sake of simplicity, we suppose that
lists of tasks of different houses are disjoint. We denote by
$Tasks(TL_i)=\{t_1,t_2,\ldots,t_k\}$, the tasks in the task list of
house $H_i$. The goal of the houses is to complete each task within
the given preferred interval and minimise
\begin{enumerate}
\item the overall imported energy consumption of all the houses;
\item as well as the bills of each house (the billing functions is
  described hereunder).
\end{enumerate}

\begin{example}
  Consider the scenario with two houses $H_1$ and $H_2$. Let consider
  only two time intervals and the production function $Prod(1)=4$ and
  $Prod(2)=2$. Thus at first interval both houses produce $4$ units of
  energy and at second interval both houses produce $2$ units of
  energy. Let the tasklist of $H_1$ be $(\langle t_1,[1,2]\rangle)$
  and that of $H_2$ be $(\langle t_2,[1,2]\rangle)$ where the energy
  required for tasks are $E_T(t_1) = 4$, $E_T(t_2)=5$.
\end{example}

\subsection{Concurrent MCR game to minimise the energy}

We first consider our primary goal being to minimise the amount of
imported energy used during the day. To model this situation, we use
a concurrent MCR game $\game$ with $N$ players representing the local
controllers of each house, as follows:
\begin{itemize}
\item $\vertices = [1,96]\times(\prod_{i\in H}2^{Tasks(TL_i)})$
  contains the current timeslot and the set of tasks already performed
  in the past;
\item $\Targets = [1,96]\times(\prod_{i\in H}Tasks(TL_i))$ describe
  that every task has been performed;
\item $A_i = Tasks(TL_i)$ is the set of tasks, for all players $i$;
\item $\edges = \{((d,p),(d+1,p'))\}$ with $p\subseteq p'$;
\item
  $\Next((d,p),(p_1,p_2,\ldots,p_N))=(d+1,p\cup p_1\cup p_2\cup \cdots
  \cup p_N)$
  if all tasks of $p_i$ are associated to an interval including $d$ in
  the task list $TL_i$ (other actions are not fireable);
\item $\edgeweights_i$ is defined as $E_T$ for all tasks of house
  $H_i$ performed in the current time, while taking into account the
  solar energy production, i.e.\
  $\edgeweights_i((d,p),(d+1,p')) = \sum_{t\in (p'\setminus p) \cap
    Tasks(TL_i)} E_T(t) - prod(d)$.
  A negative weight implies a use of energy produced outside the house
  (either by other houses or outside the local grid), while positive
  weight induces an excedent of solar energy in the house.
\end{itemize}
Note that by construction $\game$ is an acyclic graph (always
incrementing the interval component of the vertex). We will consider
thereafter only this game starting in the initial vertex
$v_0 = (1,\emptyset,\ldots,\emptyset)$.

For each edge, the sum of the weights incurred by all the houses
represent the amount (positive or negative) of solar energy excedent
after the perfomance of all the tasks of the current time. Since we
want to reduce the amount of energy bought from outside the local
grid, we use as a global weight function the negative part of this sum
of weights:
$\edgeweights ((d,p),(d+1,p')) = \min\left(0,\sum_{t\in p'\setminus p}
  E_T(t) -N\times prod(d)\right)$.
A negative weight implies a use of energy produced outside the local
grid, while a null weight induces an excedent of solar energy in the
local grid.

We decide first to interpret the previous game as a one-player game,
by supposing that all houses play in coalition to achieve the common
target of finishing all the tasks within the given interval and
reducing the usage of non-solar energy. This is a one-player MCR game
with the weight function $\edgeweights$. The coalition strategy
obtained will actually be a schedule for the tasks of each house
respecting the intervals that minimises the amount of non-solar energy
used (or even maximise the solar energy excedent produced by the local
grid to be exported).

\begin{example}
  \label{opt}
  For the example developed above, the optimal schedule in the
  coalition game is to perform task $t_2$ at interval 1 and perform
  task $t_1$ at interval 2 in which case no energy from outside is
  required. On the contrary, an excedent of 3 units of energy is
  produced which can sold to the external grid causing lowering of
  electric bill.
\end{example}

\subsection{Billing function}

Even though the schedule obtained from $\game$ gives the optimal use
of non-solar energy and hence a priori low billing costs as a whole
for the houses, the bill obtained may not be favourable for the houses
taken individually. Thus, all houses might not have a strong incentive
to collaborate to the \emph{common good}. We start by defining
properly the billing function we use in our model. 

Given a tuple of sets of tasks performed by each house at any specific
time point $d\in\{1,\ldots,96\}$, we will compute the bill incurred by
house $H_i$ on the interval $[d,d+1)$. The total bill for $H_i$ would
then be the sum of bill incurred by this house for each interval of
the day.

Consider a tuple of set of tasks performed by all the houses at a
specific time point $d$,
$T_P = (\langle
t^1_1,t^1_2,\allowbreak\cdots,t^1_{k_1}\rangle,\cdots,\langle
t^N_1,t^N_2,\cdots,t^N_{k_N}\rangle)$.
We denote the tasks performed by $H_i$ as $Tasks_i(T_P)$. Let the
price of buying energy from other houses be $P_{in}$ and the price of
buying energy from outside be $P_{out}$. The energy produced by each
house is given by $prod(d)$.  Now, for each house $H_i$, the excess
energy used by the house is given by
$\sum_{t\in Tasks_i(T_P)} E_T(t)-prod(d)$. Thus, the total energy
bought by all the houses individually (either from the local grid or
from the outside) is
$Tot_C = \sum_i \max\left(0,\sum_{t\in Tasks_i(T_P)}
  E_T(t)-prod(d)\right)$.
On the other hand, the energy bought (negative or positive) from
outside the grid is
$Tot_O = \sum_i \sum_{t\in Tasks_i(T_P)} E_T(t) - N\times prod(d)$.
The total bill for all the houses is then
$B_{Tot} = (Tot_C - Tot_O)\times P_{in} + Tot_O\times P_{out}$. Since
each house pays its own share of this total bill, the price that will
be billed to house $H_i$ is
$\edgeweights^B_i(T_P) = B_{Tot}/Tot_C \times \sum_{t\in Tasks_i(T_P)}
E_T(t)-prod(d)$.
Note that if a house produces more energy at a specific interval than
it consumes, the bill is negative signifying income from selling the
excess solar energy.

Now that we have the billing function fixed, we can present the
example where the optimal energy schedule might not give the minimum
bill for an individual house.

\begin{example}
  The optimal schedule for the total energy presented in the previous
  example (Example~\ref{opt}) is not optimal with respect to the bill paid by house
  $H_1$. For example, if $H_1$ performs task $t_1$ at interval 1, it
  does not have to pay anything. Whereas, with optimal scheduling,
  $H_1$ has to pay for two units of energy to $H_2$ and receives the
  price of only one unit of energy from $H_2$.
\end{example}

Thus, our next goal will be to modify the weights of the game $\game$
to take into account the bill rather than the energy. The new weight
function is now given by
$\edgeweights_i ((d,p),(d+1,p')) = \edgeweights^B_i(T_P)$ where $T_P$
is the list of tasks performed in $p'\setminus p$. We call $\game'$
this new game.
The hope is to find that the need for the households to minimise their
utility bill is an \emph{incentive} to minimise the global energy
consumption from the grid (thereby encouraging sharing of locally
produced energy). More formally, we need to compare the energy
consumed by a Nash equilibrium of $\game'$ to the optimal energy
consumption found in the optimal coalition strategy of $\game$.

As the game is concurrent there is in general no Nash equilibrium.
Therefore we start by transforming $\game'$ in a turn-based game
$\game'_t$, adopting a round-Robin policy for the choice of actions.
This can be achieved by enhancing the set of vertices with
$\{1,\ldots,N\}$, and decomposing an edge into a sequence of $N$
edges, where each house now plays in turns. In the last step, we have
all the information to compute the bill for each house.
%
%
Since the game $\game'_t$ is acyclic, there are only finitely many
plays, thus their costs for each players are bounded. As a consequence
of Theorem~\ref{thm:pureNash}, we know that there exists (pure) Nash
equilibria in $\game'_t$.\footnote{Notice that we could also obtain
  this result directly from the fact that every acyclic turn-based
  game has a Nash equilibrium.} Thus we can follow the
heuristic for constructing Nash equilibria presented in Section~\ref{folk} to
construct the Nash equilibrium strategy profile.  We construct
coalition two-player MCR games $\game'_i$ for each house $H_i$, where
$H_i$ plays in order to minimise its bill against the coalition of all
other houses.  Solving every such game $\game'_i$, we obtain the
optimal strategy $\strategy_i$ for each house $H_i$.  In addition to that, we follow the construction by detecting when a player deviates from its optimal
strategy and then changing other players' strategy to a punishment
strategy.  

From the point of view of the case study, even though the
strategies $(\strategy_i)_{1\leq i\leq N}$ are generated by the global
controller, they are executed by local controllers and thus, each
house can not detect whether some other house has deviated from its
optimal strategy or not. Hence, for our case, we only take the
strategy profile $(\strategy_i)_{1\leq i\leq N}$ (without the
deviation punishment) and, while computing the bill, we add the
provision for the global controller to add a penalty to the bill. This
is done by modifying the weight function to incorporate such changes:
we add an extra integer to the bill of house $H_i$ that is equal to
the minimum bill that can be ensured by $H_i$ according to the
strategy $\strategy_i$ whenever $H_i$ deviates from
$\strategy_i$. This ensures that any deviation from $H_i$ will result
in at least twice the minimal bill that can be ensured by $H_i$.


\section{Implementation}

We implemented the model using PRISM. PRISM has introduced a module for
solving \emph{(turn-based) Stochastic Multi-Player Games (SMG)}. We
use this module in order to solve different non-stochastic games
and extract optimal cost strategies out of them. The PRISM module is
also used to check the performance (consumption, wastage and bill) of
a strategy over an instance of the game.

We have first implemented the one player game version of an instance
where all houses play in coalition towards the common goal of
maximising the utilisation of solar energy. Here, the behaviour of
each house is modeled using a module in the PRISM representation. Each
module contains the constraints of the houses with respect to tasks as
transitions. The favorable interval of the task is denoted as guards
on the transition and the energy cost for the task is reflected using
an update to the global energy variable. We solve the game and obtain
a bound of the maximum possible utilisation of the solar energy among
all the houses. Note that, as shown by the example in the previous
section, this schedule does not ensure that the bill paid by each of
the houses is minimum. We allow PRISM to solve such a one-player game
to figure out the minimum possible collective energy requirement of
the houses ($E_{min}$).

Next we have implemented the methodology with multi-player turn-based
MCR games. Recall that the houses do not have information about
consumption and requirement of energy by other houses. The natural way
of modelling such scenarios is through concurrent games where each
player plays a move without the knowledge of other players
moves. Since PRISM can handle only turn-based games, we try to
implement a random order among the houses at each step of the game. We
then compute the separate games $\game'_i$ for each house $H_i$ and
find optimal strategy $\strategy_i$ for house $H_i$ such that the bill
for $H_i$ is minimised ($bill_i$). Even though generating strategy is
included in PRISM, it does not allow storing the strategy output in a
proper format for further usage from the command interface. We
modified it to include that property. The outcome of this strategy
profile $(\strategy_i)_{1\leq i\leq N}$ is then used to compute the
final strategy for the controller. Finally, we formulate another game
where any deviating move by house $H_i$ from $\strategy$ contains a
modification of the billing function of $H_i$ as an addition of
integer value equal to $bill_i$. At the end, the final strategy is
loaded in PRISM and the values (energy consumption, billing\ldots)
corresponding to the strategy are computed. The final game with the
strategy profile $(\strategy_i)_{1\leq i\leq N}$ again results in
various different values for total collective energy consumption, and
bills for each house $H_i$ for completing all the tasks. These values
are compared with the original game to compare the performance of the
strategy profile. The table below shows the result for different
numbers of houses and tasks. For each such pair, we have taken 10
examples and presented the average of values obtained. The table
represents the average difference of bills (in percent) between two strategies -
one where the houses collectively reduce the total energy consumed in
coalition and the other where the houses minimizes their own bill. Note that
the bill is computed for each house by taking into account the cost of
excess energy used in any interval. However, the amount each
house gets from excessive production of energy is not accounted for in the bill.
\begin{table}[t]
  \centering
  \footnotesize
    \begin{tabular}{|c|c|c||c|c|}
\hline
    Houses&Tasks&Number of cases&Total energy difference&Average bill difference\\ \hline\hline
    2&3&10&0.0&-8.08\\ \hline
    2&4&10&0.0&-17.15\\ \hline
    3&2&10&0.0&-13.07 \\ \hline
    3&3&10&0.0&-29.73 \\ \hline
    4&2&10&0.0&-14.89 \\ \hline
    \end{tabular}
  \caption{Results of the implementation over the case study}\label{table}
\end{table}
As shown in Table~\ref{table}, the collective energy with the strategy
profile obtained, remains the same as the minimum energy required to
complete all the tasks. Moreover, the result shows that on average
there is a decrease in bill paid by each house in the case where every
house follows the strategy profile and does not deviate from it. This
also shows that there is (hopefully) less inclination towards
deviating from the suggested strategy by each house.


%

\bibliographystyle{eptcs}
\bibliography{biblio}
\end{document}